\documentclass[conference]{IEEEtran}

\usepackage{cite}

\usepackage{graphicx,color,epsfig,rotating,subfigure}
\usepackage{amsfonts,amsmath,amssymb}
\usepackage{algorithm}
\usepackage{caption}
\usepackage{subfigure}
\usepackage{algpseudocode}
\usepackage{placeins}
\usepackage{psfrag, graphicx}

\graphicspath{{figures/}}
\usepackage{booktabs}
\usepackage{longtable}
\usepackage{tabu}
\usepackage{multirow}
\usepackage{multicol}
\usepackage{accents}

\newtheorem{definition}{Definition}
\newtheorem{theorem}{Theorem}

\newtheorem{proof}{Proof}
\newtheorem{remark}{Remark}

\newtheorem{example}{Example}

\def\beq{\begin{equation}}
\def\eeq{\end{equation}}
\def\beqa{\begin{eqnarray}}
\def\eeqa{\end{eqnarray}}
\def\beqan{\begin{eqnarray*}}
\def\eeqan{\end{eqnarray*}}

\def\EE{{\mathbb{E}}}
\def\PP{{\mathbb{P}}}

\def\Wsf{ W}
\def\Xsf{ {\sf X}}

\def\Xc{\mathcal{X}}
\def\Yc{\mathcal{Y}}

\def\Dc{\mathcal{D}}

\def\dbf{\mathbf{d}}

\def\Rscr{\boldsymbol\varrho}
\def\R{\rho}
\def\GWregion{\mathfrak S_{GW}} 
\def\regopGW{\mathfrak R^{*}_{GW}}         

\def\RGWstarR{R_{GW}^*(M,\Rscr)}
\def\RGWstar{R_{GW}^*(M)}   

\def\Rstar{R^*(M)}              

\def\Rlb{{R}^{LB}(M)}         

\def\Rub{{R}_{GW}^{UB}(M)}         

\def\RubR{{R}_{GW}^{UB}(M,\Rscr)}         

\def\regop{\mathfrak R^{*}}         

\IEEEoverridecommandlockouts
\begin{document}

\title{\Huge Broadcast Caching Networks with Two Receivers and Multiple Correlated Sources}
\author{Parisa Hassanzadeh, Antonia M. Tulino, Jaime Llorca, Elza Erkip
\thanks{This work has been supported by NSF  grant \#1619129.}
\thanks{P. Hassanzadeh  and  E. Erkip are with the ECE Department of New York University, Brooklyn, NY. Email: \{ph990, elza\}@nyu.edu}
\thanks{J. Llorca  and A. Tulino are with Nokia Bell Labs, Holmdel, NJ, USA. Email:  \{jaime.llorca, a.tulino\}@nokia.com}
\thanks{A. Tulino is with the DIETI, University of Naples Federico II, Italy. Email:  \{antoniamaria.tulino\}@unina.it}
}

\maketitle

\begin{abstract}
The correlation among the content distributed across a cache-aided broadcast network can be exploited to reduce the delivery load on the shared wireless link. This paper considers a two-user three-file network with correlated content, and studies its fundamental limits for the worst-case demand. A class of achievable schemes based on a two-step source coding approach is proposed. 
Library files are first compressed using Gray-Wyner source coding, and then cached and delivered using a combination of correlation-unaware cache-aided coded multicast schemes. The second step is interesting in its own right and considers a multiple-request caching problem, 
whose solution requires coding in the placement phase. 
A lower bound on the optimal peak rate-memory trade-off is derived, which is used to evaluate the performance of the proposed scheme. It is shown that for symmetric sources the two-step strategy achieves the lower bound for large cache capacities, and it is within half of the joint entropy of two of the sources conditioned on the third source for all other cache sizes.
\end{abstract}

\section{Introduction}~\label{sec:Introduction}
Coded multicast transmissions can significantly increase the capacity of wireless access networks by leveraging the broadcast nature of the wireless channel and the content that is distributed across the network \cite{maddah14fundamental}. In this work, we consider the cache-aided coded multicast problem, extensively studied in the literature for a library composed of independent content \cite{maddah14fundamental,ji15order}, in the setting of a content library composed of correlated files, as investigated in \cite{timo2016rate,ITW2016,ISTC2016,ISIT2017,yang2017centralized}. Rate-memory-distortion trade-offs in a single-receiver network were studied in \cite{timo2016rate}, while \cite{ITW2016} and \cite{ISTC2016} considered more general networks with multiple receivers and provided schemes that exploit the correlation among the content files during the caching phase and delivery phase, respectively.  
This paper builds upon the results of \cite{ISIT2017}, which provides an information theoretic analysis of the peak delivery rate in a two-receiver two-file network, by considering a setting with two receivers and three correlated files. We explore the content correlations by first compressing the correlated library based on the {\em Gray-Wyner network} \cite{gray1974source}, and then treating the resulting encoded content as independent files. The achievable strategy can be generalized to an arbitrary number of files, but in this paper we focus on the three file scenario since it captures the essence of caching in broadcast networks with multiple files. 
In fact, the exponential complexity of Gray-Wyner source coding with the number of files makes the overall characterization with large number of files exceedingly hard without providing additional insight. 
Concurrent work in \cite{yang2017centralized} studies the second step of the proposed scheme for arbitrary number of files and users.

The main contributions of this paper, beyond the results presented in \cite{ISIT2017}, are as follows: 
\begin{itemize}
	\item We consider a broadcast caching network with {\em three} correlated files, and characterize the optimal or near-optimal peak rate-memory trade-off. 
	
	\item The proposed correlation-aware achievable scheme for the three file scenario, as in the two-file case studied in  \cite{ISIT2017}, is a two-step scheme, for which the second step results in multiple per-user requests. However, unlike the two-file case, receiver requests are not symmetric across the content generated 
	by the first step. The aforementioned asymmetry results in an interesting three-file two-request problem, where prefetching coded content is key for optimality in the low cache capacity regime, as opposed to the single request setting where caching uncoded content is sufficient \cite{yu2017characterizing}. 
				  


	\item We discuss the optimality of the proposed two-step correlation-aware scheme by comparing its achievable rate with the lower bound on the optimal rate-memory trade-off. We identify a set of operating points in the Gray-Wyner region, for which the two-step scheme is optimal over a range of cache capacities, and discuss how far away it is from optimal for other capacities.	
\end{itemize}

The paper is organized as follows. Sec.~\ref{sec:problem} presents the system model and problem formulation. Sec.~\ref{sec: GW scheme} proposes a class of achievable schemes based on Gray-Wyner compression, which converts the original problem into a  multiple-request caching problem, studied in Sec.~\ref{sec: achievable GW-CACM}. The main results of the paper are provided in Sec.~\ref{sec:optimality}, and the paper is concluded in Sec.~\ref{sec:conclusion}.

\section{System  Model and Problem Statement}~\label{sec:problem}
 We consider a broadcast caching network composed of one sender (e.g., base station)
 with access to a library with three uniformly popular files generated by an $3$-component discrete memoryless source (3-DMS). The 3-DMS model $\Big(\Xc_1 \times \Xc_2 \times\Xc_3, \, p(x_1,x_2,x_3)\Big)$ consists of $3$ finite alphabets $\Xc_1,\Xc_2,\Xc_3$ and a joint pmf  $p(x_1,x_2,x_3)$ over $\Xc_1 \times \Xc_2\times\Xc_3$. The 3-DMS generates an i.i.d. random process $\{\Xsf_{1i},\Xsf_{2i},\Xsf_{3i}\}$ with $(\Xsf_1, \Xsf_2, \Xsf_3)\sim p(x_1,x_2,x_3)$.  For a block length $F$, library file $j \in \{1,2,3\}$ is represented by a sequence $\Xsf_j^F = (\Xsf_{j1},\dots,\Xsf_{jF})$, where $\Xsf_j^F\in \Xc^F_j$. The sender communicates with two receivers, $\{r_1,\, r_2\}$, over a shared error-free broadcast link. Each receiver is equipped with a cache of size $FM$ bits, where $M$ denotes the (normalized) cache capacity.
 
  The network operates in two phases: a caching phase taking place at network setup, in which caches are populated with content from the library, followed by a delivery phase where the network is used repeatedly in order to satisfy receiver demands. The overall scheme is referred to as a {\em cache-aided coded multicast scheme} (CACM).   A CACM scheme consists of the following components:
 \begin{itemize}
 	\item {\textbf{Cache Encoder:}} During the caching phase, the cache encoder designs the cache content of receiver $r_k$ using a mapping $$f^{\mathfrak C}_{r_k}:\Xc_1^F\times\Xc_2^F\times\Xc_3^F \rightarrow [1: 2^{FM}).$$ 
 	The cache configuration of  receiver $r_k$  is denoted  by 
 	$$Z_{r_k} = f^{\mathfrak  C}_{r_k}\Big(\{X_1^F,X_2^F, X_3^F\}\Big).$$

 	\item{\textbf{Multicast Encoder:}} During the delivery phase, each receiver requests a file from the library. The demand realization, denoted by $\dbf = (d_{r_1},d_{r_2}) \in \Dc \equiv \{1,2,3\}^2$, where $d_{r_k} \in\{1,2,3\}$ denotes the index of the file requested by receiver $r_k$, is revealed to the sender,
 	which then uses a fixed-to-variable mapping 
 	$$f^{\mathfrak M}:{\mathcal D} \times [1: 2^{FM})^2  \times \Xc_1^F\times\Xc_2^F\times \Xc_3^F  \rightarrow \Yc^\star$$ 
 	to generate and transmit a multicast codeword 
 	$$Y_{\dbf} = f^{\mathfrak  M}\Big(\dbf,\{Z_{r_k}\}, \{X_1^F,X_2^F,X_3^F\}\Big)$$ 
 	over the shared link.\footnote{We use $\star$ to indicate variable length.}

 	\item{\textbf{Multicast Decoders:}} Each receiver $r_k$ uses a mapping 
 	$$g^{\mathfrak  M}_{r_k} : \Dc \times \Yc^\star \times [1: 2^{FM}) \rightarrow \Xc_{d_{r_k}}^F$$ 
 	to recover its requested file,
 	$X_{d_{r_k}}^F$, using the received multicast codeword and its cache content as 
 	$$\widehat{X}_{d_{r_k}}^F = g^{\mathfrak  M}_{r_k} (\dbf, Y_{\dbf},Z_{{r_k}}).$$
 	
 \end{itemize}

 The worst-case probability of error of a CACM scheme is given by
 \begin{align} \label{perr}
 & P_e^{(F)} = \max_{\dbf} \;\max_{r_k}\;  \PP\left(\widehat{X}_{{d_{r_k}}}^F  \neq X_{{d_{r_k}}}^F \right).
 \end{align}

  In this paper, we consider the peak multicast rate, $R^{(F)}$, which corresponds to the worst-case demand,
 	\begin{equation} \label{peak-rate}
 	R^{(F)} =  \max_{\dbf \in \mathcal D} \; \frac{\EE[L(Y_{\dbf})]}{F}, 
 	\end{equation}
 	where $L(Y)$ denotes the length (in bits) of codeword $Y$, and the expectation is over the library files.

 	\begin{definition} \label{def:achievable-peak}
 		A peak rate-memory pair $(R,M)$ is {\em achievable} if there exists a sequence of CACM schemes for cache capacity $M$ and
 		increasing file size $F$, such that 
 		$$\lim_{F \rightarrow \infty} P_e^{(F)} = 0 \notag,$$ 
 		and 
 		$$\limsup_{F \rightarrow \infty}	R^{(F)} \leq  R.$$ 
 	\end{definition}

 	\begin{definition} \label{def:infimum-peak}
 		The peak rate-memory region, $\regop$, is the closure of the set of achievable peak rate-memory pairs $(R,M)$, and the optimal peak rate-memory function, $\Rstar$, is
 		$$\Rstar= \inf \{R:  (R,M) \in  \regop\}.$$
 	\end{definition}

\section{Proposed Correlation-Aware Scheme }\label{sec: GW scheme}
We propose a class of CACM schemes based on a two-step lossless source coding setup, as depicted in Fig.~\ref{fig:model}. The first step involves an extension of the two-component Gray-Wyner network \cite{gray1974source} to multiple sources, and the second step is a  lossless correlation-unaware multiple-request source coding scheme with distributed side information.  We refer to this scheme as {\em Gray-Wyner Cache-Aided Coded Multicast} (GW-CACM). 
The GW-CACM scheme exploits the correlation among the library content by first compressing the library using the three-file Gray-Wyner encoder depicted in Fig.~\ref{fig:GW}(a) and explained in detail in Sec.~\ref{subsec: GW network}. The three files are encoded into seven descriptions, such that: i) three of the descriptions contain information exclusive to only one file, and ii) the remaining descriptions comprise information common to more than one file.
\begin{figure}[t!] \centering
	\includegraphics[width=1\linewidth]{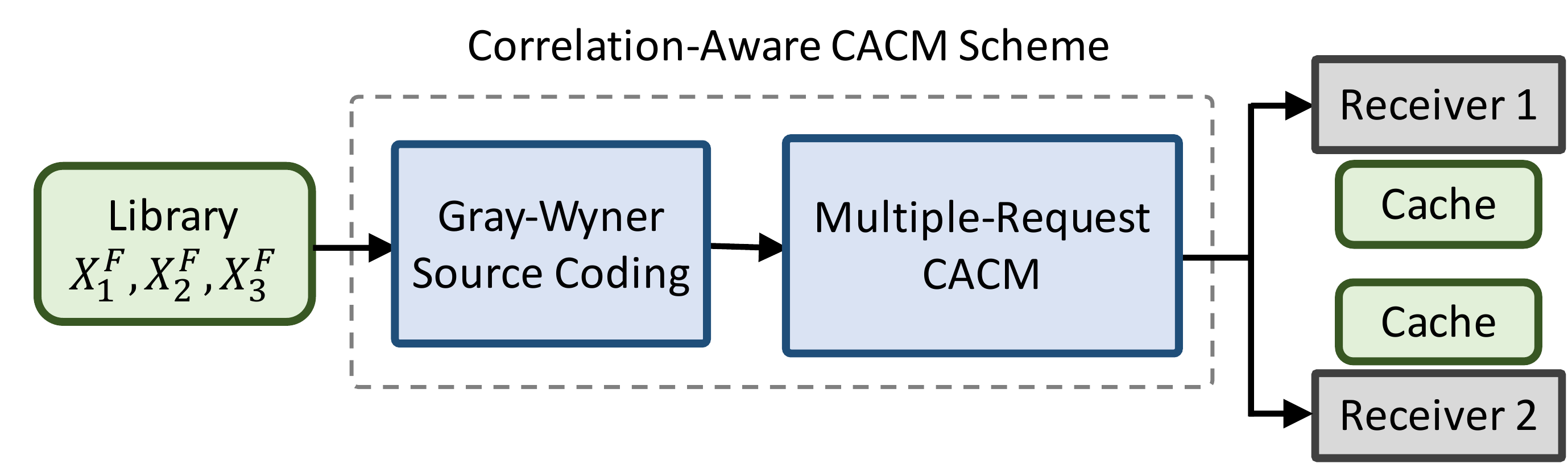}
	\caption{Two-step correlation-aware scheme.}
	\label{fig:model}
\end{figure}

This initial Gray-Wyner source coding step transforms a caching problem in a network with correlated content and receivers requesting only one file, into a caching problem with a larger number of files where receivers request multiple descriptions to recover the desired content. We assume that the CACM scheme in the second step is agnostic to the correlation among the content generated by the Gray-Wyner encoder, i.e., the second step is correlation-unaware. The two steps are jointly designed to optimize the performance of the overall scheme. Before formally describing the GW-CACM scheme, we briefly review the three-file Gray-Wyner network.

\subsection{Gray-Wyner Network}\label{subsec: GW network}
Gray-Wyner source coding was first introduced for two files in \cite{gray1974source}, in which a 2-DMS  is represented by one common description, and two private descriptions such that each of the files can be losslessly recovered from the common description and one of the private descriptions, asymptotically, as the file size $F\rightarrow \infty$. In \cite{gray1974source}, Gray and Wyner fully characterized the rate region for lossless reconstruction of both files. The Gray-Wyner network can be extended to three files such that the Gray-Wyner encoder observes a 3-DMS $(\Xsf_1, \Xsf_2,\Xsf_3 )$, and communicates $\Xsf_i$ to decoder $i\in\{1,2,3\}$. As depicted in Fig.~\ref{fig:GW}, the encoder is connected to the decoders through three types of error-free links with finite rate. There is one common link connecting the encoder to all three decoders, there are three links common to any two of the decoders, and finally there are three private links connecting the encoder to each decoder. Description $\Wsf_{\mathcal A} \in [1:2^{F\R_{\mathcal A}})$, is communicated to all decoders $i \in { \mathcal A} \subseteq\{1,2,3\}$, such that\footnote{With an abuse of notation, the subscripts of $\Wsf$ and $\R$ denote sets.}
\begin{itemize}
	\item	$\Wsf_{123}$ $\in [1:2^{F\R_{123}})$,   
	\item	$\Wsf_{12} \in [1:2^{F\R_{12}})$, $\Wsf_{13} \in [1:2^{F\R_{13}})$,  $\Wsf_{23} \in [1:2^{F\R_{23}})$, 
	\item	$\Wsf_1 \in [1:2^{F\R_1})$, $\Wsf_2 \in [1:2^{F\R_2})$, and  $\Wsf_3 \in [1:2^{F\R_3})$.
\end{itemize}
The Gray-Wyner region, $\GWregion$, is represented by the set of all rate-tuples
 $$\Rscr =\Big(\R_{123},\R_{12},\R_{13},\R_{23},\R_{1},\R_{2},\R_{3}\Big),$$
  for which any file $X_i^F$, $i\in\{1,2,3\}$,  can be losslessly reconstructed from the descriptions $\Big\{\Wsf_{123}, \Wsf_{ij}, \Wsf_{ik} , \Wsf_i \Big\}$ with $k, j \neq i$, asymptotically, as $F \rightarrow \infty$.

While the generalization of the Gray-Wyner network to multiple files has been studied in a number of papers, \cite{tandon2011multi,liu2010common,viswanatha2012subset}, the optimal characterization of the rate region for generic sources is nontrivial and is not known.

\begin{figure}
	\centering
		\includegraphics[width=0.95\linewidth]{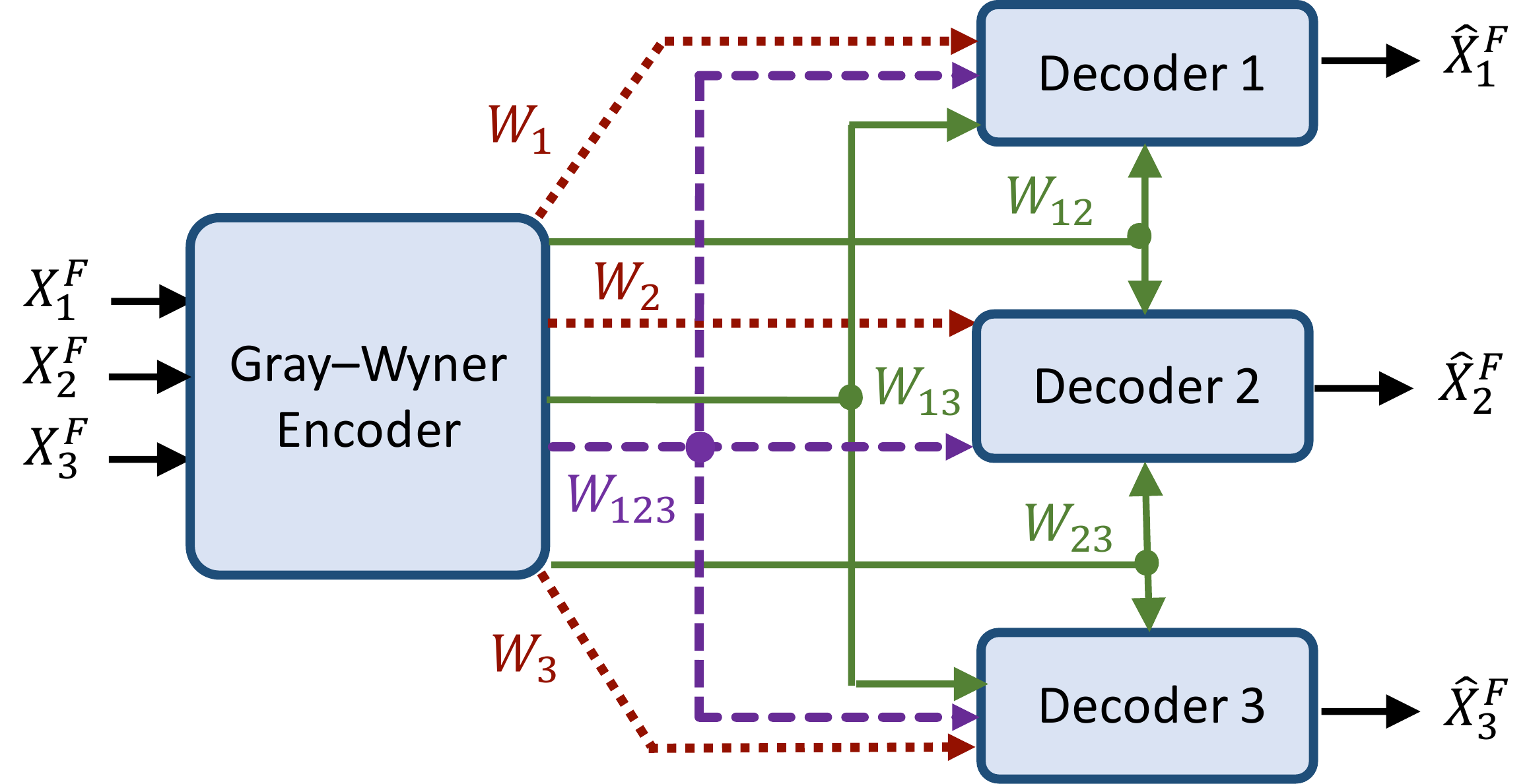} 
	\caption{Three-file Gray-Wyner network.}
	\label{fig:GW}
\end{figure}

\subsection{Gray-Wyner CACM  Scheme}\label{subsec: general GW-CACM} 
The Gray-Wyner network converts the library into a set of descriptions, each of which is required for the lossless reconstruction of one or more of the original files. 
For a given rate-tuple $\Rscr  \in \GWregion$, a GW-CACM scheme consists of:
\begin{itemize}
	\item {\bf Gray-Wyner Encoder}: Given the library $\{X_1^F$, $X_2^F$, $X_3^F\}$,
	the Gray-Wyner encoder at the sender computes descriptions $\{\Wsf_{s} :   s\in \mathcal S\}$, where $\mathcal S$ is the set of all nonempty subsets of $\{1,2,3\}$, using a mapping 
	$${f}^{GW}: \Xc_1^{F}\times\Xc_2^{F}\times\Xc_3^{F} \rightarrow \prod\limits_{s \in \mathcal S }  \Big[1:2^{F\R_s}\Big)  .$$	
	
	\item {\bf Multiple-Request Cache Encoder}: Given the compressed descriptions, the correlation-unaware cache encoder at the sender computes the Gray-Wyner based cache content at receiver $r_k$, as
	$$Z_{r_k} =  f_{r_k}^{{\mathfrak C_{GW}}}\Big (  \{\Wsf_{s} :   s\in \mathcal S\}  \Big).$$

	\item {\bf Multiple-Request Multicast Encoder}: For any demand realization $\dbf\in \mathcal D$ revealed to the sender, the correlation-unaware  Gray-Wyner-based multicast encoder generates and transmits the multicast codeword
	$$Y_\dbf^{{{GW}}} = f^{{\mathfrak M_{GW}}}\Big(\dbf,\{Z_{r_1},Z_{r_2}\}, 	\{\Wsf_{s} :   s\in \mathcal S\} \Big).$$
	
	\item {\bf Multicast Decoder}: Receiver $r_k$ decodes the descriptions corresponding to its requested file as
	$$\Big\{ \widehat{\Wsf}_s: \, s\in  \mathcal S_{d_{r_k}}  \Big\} = g^{{\mathfrak M_{GW}}}_{r_k}\Big(\dbf, Y_\dbf^{{{GW}}}, Z_{r_k}\Big),$$
	where 
	$$\mathcal S_{d_{r_k}} =  \Big\{ \{1,2,3\}, \{d_{r_k}, i\} ,\{ d_{r_k}, j \} , \{d_{r_k} \} \Big\},\;  i, j \neq d_{r_k}.$$
	
	\item {\bf Gray-Wyner Decoder}: Receiver $r_k $ decodes its requested file using the descriptions recovered by the multicast
	decoder, via a mapping 
	$$g^{GW}_{r_k}:  \prod\limits_{s\in \mathcal S_{d_{r_k}}  }  [1:2^{F\R_s })  \rightarrow \Xc_{d_{r_k}}^F,$$
	 as
	$$\widehat{X}_{d_{r_k}}^{F} = g^{GW}_{r_k}\Big(  \Big\{ \widehat{\Wsf}_s: \, s\in  \mathcal S_{d_{r_k}}  \Big\} \Big).$$
\end{itemize}

As in \eqref{peak-rate}, the peak GW-CACM multicast rate is
\begin{align} \label{average-rate-GW}
& R_{GW}^{(F)}(\Rscr) = \max_{\dbf \in \mathcal D} \;  \frac{\EE[L(Y^{GW}_{\dbf})]}{F}, \notag
\end{align}
where we explicitly show the dependence on the rate-tuple $\Rscr$.

For the class of GW-CACM schemes, since $\Rscr \in \GWregion$,  and $ \{ {\Wsf}_s: \, s\in  \mathcal S_{d_{r_k}}  \}  $ is a Gray-Wyner description of $X_{{d_{r_k}}}^F$ with $d_{r_k}\in\{1,2\}$,  in order to have a sequence of admissible GW-CACM schemes, i.e., a sequence of schemes such that
$\lim_{F \rightarrow \infty}  P_e^{(F)}= 0$,  with $P_e^{(F)}$ as defined in \eqref{perr},
we only need
\begin{align}
\lim_{F \rightarrow \infty} \max_{\dbf\in \mathcal D} \;\max_{r_k}\; 
\PP\Big(   \widehat{\Wsf}_s \neq   \Wsf_s,\, \forall   s\in  \mathcal S_{d_{r_k}}    \Big) = 0. \notag
\end{align}
Hence, a peak $\Rscr$-rate-memory pair $(R,M)$ is achievable if  there exists a sequence of admissible GW-CACM  schemes with rate $R_{GW}^{(F)}(\Rscr)$  such that  $\lim\sup_{F\rightarrow \infty} R_{GW}^{(F)}(\Rscr) \leq R$.

In line with Definition \ref{def:infimum-peak}, for the class of GW-CACM schemes, for a given rate-tuple $\Rscr \in\GWregion$, the peak $\Rscr$-rate-memory region, $\regopGW(\Rscr)$, is the closure of the set of all achievable peak $\Rscr$-rate-memory pairs $(R,M)$, and the peak $\Rscr$-rate-memory function, $\RGWstarR$ is 
\begin{align}
&\RGWstarR= \inf \{R:  (R,M) \in  \regopGW(\Rscr)\}. \notag
\end{align}
Finally, the peak GW-rate-memory function of the GW-CACM scheme, $\RGWstar$, is defined as
\begin{align}
&\RGWstar=\inf\{\RGWstarR:\,\Rscr\in\GWregion \}. \notag 
\end{align}
 
In the following sections, we analyze the rate-memory trade-off for the class of achievable GW-CACM schemes, and discuss its optimality.

\section{Multiple-Request CACM} \label{sec: achievable GW-CACM}
In this section,  we focus on the second step of the GW-CACM scheme depicted in Fig.~\ref{fig:model}, namely the multiple-request CACM, and  propose an achievable strategy, such that the Gray-Wyner encoder in the first step is restricted to operate  at a symmetric rate vector, i.e.,
$$\Big\{\Rscr\in\GWregion:\, \R_{12}=\R_{13}=\R_{23} = \R',\,  \R_{1}=\R_{2}=\R_{3} = \R  \Big\}, $$ 
and $\R_0$ is used to denote the rate of description $\Wsf_{123}$. 

The multiple-request CACM scheme arranges the seven descriptions generated by the Gray-Wyner encoder into three groups, referred to as {\em sublibraries}, namely $L_1$, $L_2$ and $L_3$.
Sublibrary $L_\ell$ contains the descriptions that are communicated to any subset of  size $\ell\in\{1,2,3\}$ of the Gray-Wyner decoders, i.e., each description in $L_\ell$ contains information common to $\ell$ files. Sublibrary $L_3=\{\Wsf_{123} \}$ contains information common to all three files and is referred to as the {\it common-to-all} sublibrary, $L_2 = \{\Wsf_{12}, \Wsf_{13}, \Wsf_{23} \}$ is the {\it common-to-two} sublibrary, and sublibrary $L_1=\{\Wsf_1 ,\Wsf_2 ,\Wsf_3\}$ contains information exclusive to each file and is referred to as the {\it private sublibrary}.  The multiple-request CACM  accounts for populating the receiver caches with content from sublibraries $L_1$, $L_2$ and $L_3$, and using the network repeatedly for different demand realizations. Each receiver's request corresponds to four descriptions, one description from $L_1$, two descriptions from $L_2$, and one from $L_3$ (hence the name multiple-request), which enable lossless recovery of its requested file. Even though the receivers request files from the original library independently and according to a uniform demand distribution, the structure of the corresponding demand in the multiple-request CACM is dictated by the collective of requested files, resulting in a non-uniform multiple-request demand that is not independent across the receivers. This is due to the fact that the descriptions requested by each receiver from sublibrary $L_2$ are determined from the entire demand vector. In the following we introduce a CACM scheme that is specifically designed for such structures of the demand. Multiple-request schemes in the literature, such as the ones in \cite{ji15efficient, ji2015caching}, in addition to applying to settings with equal-length files are designed for arbitrary demands, and therefore, result in higher load over the shared link.


The proposed multiple-request CACM scheme treats the descriptions in $L_1$, $L_2$, and $L_3$ as independent content and operates as follows: $i)$  the cache capacity $M$ is  divided among the three sublibraries, $ii)$  each sublibrary is cached independently from the others, and $iii)$ the content requested from the sublibraries is delivered without further coding across the sublibraries.  Specifically, the descriptions from each sublibrary are cached and delivered as follows: $i)$ description $\Wsf_{123}$ in $L_3$ is cached according to the Least Frequently Used (LFU)\footnote{LFU is a local caching policy that, in the setting of this paper, leads to all receivers caching the same part of the file.} strategy and delivered through uncoded  (naive)  multicasting, $ii)$ for the descriptions in $L_2$, a new two-request CACM scheme proposed in Sec.~\ref{sec: new scheme} is used, and finally $iii)$ sublibrary $L_1$ is cached and delivered according to the scheme proposed by Yu, Maddah-Ali and Avestimehr in \cite{yu2016exact}, referred to as \textit{YMA}.   
The  cache allocation among the three sublibraries, which takes on a water-filling-type solution \cite{Journal},  is the result  of an optimization aiming at minimizing the overall rate given by the sum of the rate achieved by each sublibrary based on their respective delivery mechanisms. The  proposed multiple-request CACM scheme can be described in terms of a cache encoder and a multicast encoder, whose detailed descriptions are given below.  

\begin{itemize}
	\item {\bf Cache Encoder}: The cache encoder populates the receiver caches such that: 
	\begin{itemize}
		\item If $M\in \Big[0,\, \frac{3}{2}\R' \Big)$, the descriptions in $L_1$ and $L_3$ are not cached at either receiver, and only the descriptions in $L_2$ are cached according to the caching phase of the two-request CACM scheme described in Sec.~\ref{sec: new scheme}.
		
		\item If $M\in \Big[ \frac{3}{2}\R'  ,\, \R_0+\frac{3}{2}(\R'+\R) \Big)$, receivers fill a portion equal to $\frac{3}{2}\R'$ from their caches with the descriptions in $L_2$ according to the caching strategy in the two-request CACM scheme.  The remainder of the cache, $M-\frac{3}{2}\R'$, is first allocated to caching identical bits of $L_3=\{\Wsf_{123}\}$ at both receivers, as per LFU caching, and the excess of capacity, if any, is used for storing the descriptions in $L_1$ according to YMA.

		\item If $M\in \Big[ \R_0+\frac{3}{2}(\R'+\R)  ,\, \R_0+3\R'+\frac{3}{2}\R\Big)$, a portion equal to $M-\R_0-\frac{3}{2}\R$ of the cache capacity is filled with the descriptions of $L_2$ according to the two-request CACM scheme, the common description $\Wsf_{123}$ is fully cached at both receivers, and $\frac{3}{2}\R$ of the capacity is allocated to storing the descriptions in $L_1$ according to YMA.
		
		\item If $M\in \Big[\R_0+3\R'+\frac{3}{2}\R ,\,  \R_0+3(\R'+\R)\Big]$, the descriptions in $L_2$ and $L_3$ are fully cached at both receivers, and the descriptions in $L_1$ are cached according to YMA over the remaining cache capacity $M-\R_0-3\R'$.  
	\end{itemize}
	
	\item {\bf Multicast Encoder}: The encoder transmits the descriptions in $L_1$ and $L_2$ according to the delivery phases of YMA and the two-request CACM scheme, respectively, while the portion of $\Wsf_{123}$ from sublibrary $L_3$ missing at each receiver cache is transmitted via uncoded  multicast.  
	
\end{itemize}
The rate achieved by the above CACM scheme will be provided in Theorem~\ref{thm:achievable rate peak} in Sec.~\ref{subsec: upperbound}. 

\subsection{Two-Request CACM Scheme} \label{sec: new scheme}
In this section, we explain in detail the CACM scheme used for the common-to-two sublibrary $L_2$. As mentioned in the previous section, for a given cache allocation among $L_1$, $L_2$ and $L_3$, caching and delivery of the content are done in an independent fashion across the sublibraries, i.e., there is no coding across the sublibraries in either phase.  As a result, the cache placement and delivery phase for sublibrary $L_2$ corresponds to a cache-aided broadcast network with two receivers and a library composed of three independent files, $\{W_{12},W_{13},W_{23}\}$, of equal size $\R' F$ bits, where each receiver requests two files from the library. Specifically, at each given time, the demand consists of $i)$ one file that is requested by both receivers, and $ii)$ two files, each requested only by one of the receivers.  While CACM schemes available in the literature such as YMA, where an uncoded prefetching strategy is adopted, are optimal for a single-request framework with two receivers, they fall short to achieve optimality in this multiple-request setting, and a new CACM design is needed. The following example illustrates that when receivers request a common file in addition to their distinct demands, coding in the content placement further leverages the caches for reducing the network load.
\begin{example}
Consider the demand where receiver $r_1$ requests files $W_{12}$ and $W_{13}$, and $r_2$ requests files $W_{12}$ and $W_{23}$. When the receivers are equipped with caches of capacity $M  = \frac{1}{2}\R'$, each file is split into two packets of length $\frac{1}{2}\R'$, and the receiver caches are filled as
\begin{align}
&Z_{r_1} = \{W_{12}^{(1)}\oplus W_{13}^{(1)}\oplus W_{23}^{(1)}\},\notag\\
&Z_{r_2} = \{W_{12}^{(2)}\oplus W_{13}^{(2)}\oplus W_{23}^{(2)} \}, \notag
\end{align}
where $W_{s}^{(i)}$ denotes packet $i$ of file $W_s$.
The codeword 
$$Y = \{W_{12}^{(1)},\;    W_{12}^{(2)},\;   W_{13}^{(2)},   \; W_{23}^{(1)}\}$$
 enables both receivers to losslessly recover their requested packets as follows:
\begin{itemize}
	\item[-] In addition to receiving $\{W_{12}^{(1)},\;    W_{12}^{(2)},\;   W_{13}^{(2)}  \}$, receiver $r_1$  can decode $ W_{13}^{(1)}$ by combining its cache content with the received packets $W_{12}^{(1)}$ and $W_{23}^{(1)}$.	
	\item[-] Similarly, $r_2$ receives the requested packets $\{W_{12}^{(1)},\;    W_{12}^{(2)},\;   W_{23}^{(1)}  \}$, and is also able to decode $W_{23}^{(2)}$ using its cache content and the transmitted packets $W_{12}^{(2)}$ and $W_{13}^{(2)}$.
\end{itemize}
This cache placement results in a delivery rate equal to $2\R'$, whereas an uncoded prefetching scheme, such as the one in \cite{yu2016exact}, achieves a rate of $\frac{7}{3}\R'$.
\end{example}
\vspace{5pt}
The above example provides an optimal placement and delivery strategy for the two-request  network for $M=\frac{1}{2}\R'$. Similar arguments can be made for other memory sizes, the details of which can be found in \cite{Journal}, where it is shown that the  memory-rate pairs, $(M,\,R)$,
$$\Big\{    \Big(0,\,3\R'\Big),\,  \Big(\frac{1}{2}\R', 2\R'\Big),\, \Big(\R', \frac{3}{2}\R'\Big),\, \Big(\frac{3}{2}\R',\,\R'\Big),\, \Big(3\R',\,0\Big)    \Big\},$$
are  achievable and optimal,  while the pair $(M,\,R) =( 2\R'  ,\,   \frac{2}{3}\R')$ is achievable  but its optimality is not proven yet. As in \cite{maddah14fundamental}, through memory-sharing the lower convex envelope of the points given above is achievable, resulting in the peak delivery rate $ R_{L_2}(M,\R')$, given as 
\begin{equation} 
R_{L_2}(M,\R') = \begin{cases}
3\R'-  2M ,                
& \; M \in  [0, \frac{1}{2}\R'  )  \\
\frac{5}{2}\R'-M  ,                
& \; M \in  [\frac{1}{2}\R', \frac{3}{2}\R'  )\\
2\R'-\frac{2}{3} M ,                
& \;M \in  [\frac{3}{2}\R', 3 \R' ] \notag
\end{cases}
\end{equation}

\subsection{Upper Bound on $\RGWstarR$}\label{subsec: upperbound}
The following theorem provides the delivery rate achieved by the proposed multiple-request CACM scheme, with the cache encoder and multicast encoder as described in the beginning of Sec.~\ref{sec: achievable GW-CACM}. 
\begin{theorem}\label{thm:achievable rate peak}
	Given a cache capacity $M$ and a rate triplet $\Rscr$, the peak rate achieved by the proposed multiple-request CACM  scheme is given by 
	\resizebox{.999\linewidth}{!}{
	\begin{minipage}{0.4\textwidth}
	\begin{align} 
	& \RubR  =  \notag\\ 
	& \begin{cases} 
	\R_0 + 3\R'+2\R -2M,   & M \in \Big[0,\, \frac{1}{2}\R' \Big)  \\
	\R_0 + \frac{5}{2}\R'+2\R -M, 	&M\in\Big[ \frac{1}{2}\R' ,  \, \R_0+\frac{3}{2}(\R'+\R)   \Big) \\
	\frac{2}{3}\R_0 + 2\R'+\frac{3}{2}\R -\frac{2}{3}M,   &M\in\Big[ \R_0+\frac{3}{2}(\R'+\R)  ,  \, \R_0+3\R'+\frac{3}{2}\R  \Big) \\
	\frac{1}{3}\R_0 + \R'+\R -\frac{1}{3}M,	 & M\in \Big[\R_0+3\R'+\frac{3}{2}\R ,\,  \R_0+3\R'+3\R\Big]. \notag
	\end{cases} 
	\end{align}
\end{minipage}}
\end{theorem}

\begin{proof}
	The proof of is given in \cite{Journal}.
\end{proof}

The rate expression given in Theorem \ref{thm:achievable rate peak} is the rate achieved by the proposed multiple-request CACM scheme for a given set of compression rates in the Gray-Wyner region. Naturally, the  overall scheme can be optimized over in the Gray-Wyner region such that the rate achieved in the second step is minimized. The peak rate  achieved by the overall GW-CACM scheme that uses the proposed multiple-request  CACM scheme in the second step with an optimized rate-tuple in the Gray-Wyner region, denoted by $\Rub$, is defined as 
\begin{align}
& \Rub = \inf\{  \RubR :\,   \Rscr\in\GWregion\}. \notag
\end{align}

\section{Optimality Results}~\label{sec:optimality}
	In this section, we provide a lower bound on the optimal peak rate-memory function, $\Rstar$, which is later used to evaluate the performance of the proposed scheme.


\subsection{Lower Bound on $\Rstar$}\label{subsec: optimal LB}
\begin{theorem}\label{thm:LowerBound}
	For a broadcast caching network with two receivers, cache capacity $M$, and a library composed of three files generated by the  distribution $p(x_1,x_2,x_3)$, a lower bound on $\Rstar$, the optimal peak rate-memory function, is given by 
	\begin{align}
	\Rlb&  =  \inf \bigg \{ R: \quad   \notag\\
	R & \geq   \max_{i,j }  H(\Xsf_i, \Xsf_j)-2M , \notag\\
	R & \geq \frac{1}{2} \Big(\max_{i,j}  H(\Xsf_i, \Xsf_j )-M\Big), \notag\\
	R & \geq \frac{1}{3}\Big(H(\Xsf_1, \Xsf_2 , \Xsf_3)-M\Big), \notag\\
	R & \geq \frac{1}{2}\Big( H(\Xsf_1, \Xsf_2, \Xsf_3) +   \max_{i} H(\Xsf_i) \Big) - M \bigg \}. \notag
	\end{align}
	
\end{theorem}
\begin{proof}
	The proof of is given in \cite{Journal}.
\end{proof}

\begin{remark}  
	The outer bound in Theorem \ref{thm:LowerBound}  
	improves the best known bounds for correlated sources given in \cite{lim2016information}, and when particularized to independent sources, matches the corresponding best known bound derived in \cite{yu2017characterizing}. 
\end{remark}

\subsection{Optimality of the Proposed  GW-CACM} 
The following theorem characterizes the performance of the proposed GW-CACM scheme for different regions of $M$, and delineates the cache capacity region for which the scheme is optimal or near optimal. Without loss of generality, we assume a symmetric 3-DMS such that $$H(\Xsf_1,\Xsf_2) = H(\Xsf_1,\Xsf_3) = H(\Xsf_2,\Xsf_3),$$ 
and  
$$H(\Xsf_1) = H(\Xsf_2) = H(\Xsf_3) .$$
\begin{theorem}\label{thm: optimality} 
 
	Let $(\tilde \R_0, \, \tilde\R',\, \tilde\R) $ be a symmetric rate-tuple in the Gray-Wyner region, for which $ \tilde \R_0+ 3\tilde\R'+ 3\tilde\R =H(\Xsf_1, \Xsf_2,\Xsf_3) $ and the rate $\tilde\R$ is maximized. 
	 Then, for $M\in   \Big[ \tilde \R_0+ 3\tilde\R'+ \frac{3}{2}\tilde\R, \,  H(\Xsf_1, \Xsf_2,\Xsf_3) \Big]$, the proposed GW-CACM scheme is optimal, i.e, 
	$$\Rub= \Rstar .$$
	 In addition, for $M\in \Big[ 0,\,  \tilde\R_0+ \frac{3}{2} (  \tilde\R'+ \tilde\R ) \Big) $, 
	\begin{align}
	&\Rub- \Rstar \leq  \frac{1}{2} H(\Xsf_2,\Xsf_3|\Xsf_1) -\tilde\R   , \notag
	 \end{align}
	and for $M\in \Big[  \tilde\R_0+ \frac{3}{2} ( \tilde\R'+ \tilde\R ), \,   \tilde\R_0+ 3 \tilde\R'+\frac{3}{2} \tilde\R    \Big) $, we have
	\begin{align}
	&\Rub- \Rstar \leq  \frac{1}{4} H(\Xsf_2,\Xsf_3|\Xsf_1) - \frac{1}{2} \tilde\R . \notag
	 \end{align}
\end{theorem}

\begin{proof}
	The proof of is given in \cite{Journal}.
\end{proof}

 \begin{remark}\label{rmk: not wyner common} 
	Theorem \ref{thm: optimality} suggests that operating at a point for which $\R_0+3\R'+3\R = H(\Xsf_1, \Xsf_2,\Xsf_3)$, and where the rate corrsponding to the descriptions in the private sublibrary is maximized allows us to increase the range of optimilaty in terms of the cache capacity, and also decreases the gap to optimlaity for other values of the capacity.  Analogously, it was shown in \cite{ISIT2017} that for two correlated files, it is desirable to maximize the rate of the private descriptions subject to a simialar condition on the sum rate of the entire descriptions.
	
\end{remark}

\section{Conclusions}~\label{sec:conclusion}
In this paper we have studied the fundamental rate-memory trade-off for the worst-case demand in the two-user cache-aided broadcast network with three correlated files. We have proposed a two-step achievable scheme, in which the files are first compressed using Gray-Wyner source coding, and then the encoded descriptions are treated as independent content by a multiple-request cache-aided coded multicast scheme. As a means to designing an achievable scheme for the second step of the Gray-Wyner-based scheme, we have also proposed a new scheme for a network with two users and three independent files, where each user requests two of the files. The proposed scheme uses coded placement in the caches to achieve optimality for small cache capacities. We have characterized the rate-memory trade-off in such two-step schemes and analyzed the optimality of the overall proposed scheme with respect to a lower bound on the peak delivery rate. 
 
 \bibliographystyle{IEEEtran}
 \bibliography{References}
 
%
%

\end{document}